\documentclass[11pt]{article}
\usepackage{amsmath,amssymb,amsfonts,latexsym,graphicx,amsthm}
\usepackage{fullpage,color}
\usepackage{enumerate}
\usepackage{url,hyperref}
\usepackage{comment}
\usepackage{bbm}
\usepackage{paralist}
\usepackage[linesnumbered,boxed,ruled,vlined]{algorithm2e}

\newtheorem{lemma}{Lemma}
\newtheorem{theorem}[lemma]{Theorem}
\newtheorem{definition}[lemma]{Definition}

\newtheorem{corollary}[lemma]{Corollary}

\newtheorem{claim}[lemma]{Claim}

        {\medskip}

        {\hspace*{\fill}$\Box$\par\vspace{4mm}}
        {\hspace*{\fill}$\Box$\par}

\newcommand{\etal}{{\em et al.}\ }

\newcommand{\floor}[1]{\lfloor #1 \rfloor}

\newcommand{\E}{\mathbb{E}}
\newcommand{\event}{\mathcal{E}}

\newcommand{\bad}{\mathcal{B}}
\newcommand{\ind}{\mathbf{1}}

\title{Fully Dynamic Maximal Independent Set in Expected Poly-Log Update Time}
\author{Shiri Chechik \thanks{Tel Aviv University, \href{}{shiri.chechik@gmail.com}}
	\and Tianyi Zhang \thanks{Tsinghua University, \href{}{tianyi-z16@mails.tsinghua.edu.cn}}}
\date{}

\begin{document}

\maketitle

\thispagestyle{empty}
\begin{abstract}
In the fully dynamic maximal independent set (MIS) problem our goal is to maintain an MIS in a given graph $G$ while edges are inserted and  deleted from the graph.
The first non-trivial algorithm for this problem was presented by Assadi, Onak, Schieber, and Solomon [STOC 2018]
who obtained a deterministic fully dynamic MIS with $O(m^{3/4})$ update time.
Later, this was independently improved by Du and Zhang and by Gupta and Khan [arXiv 2018] to $\tilde{O}(m^{2/3})$ update time
 \footnote{As usual $n$ is the number of vertices, $m$ is the number of edges and $\tilde{O}(\cdot)$ suppresses poly-logarithmic factors.}.
 Du and Zhang [arXiv 2018] also presented a randomized algorithm against an oblivious adversary with $\tilde{O}(\sqrt{m})$ update time.

The current state of art is by
Assadi, Onak, Schieber, and Solomon [SODA 2019]  who obtained randomized algorithms against oblivious adversary
with $\tilde{O}(\sqrt{n})$ and $\tilde{O}(m^{1/3})$ update times.

In this paper, we propose a dynamic randomized algorithm against oblivious adversary with
expected worst-case update time of $O(\log^4n)$. As a direct corollary, one can apply the black-box reduction from a recent work
by Bernstein, Forster, and Henzinger [SODA 2019] to achieve $O(\log^6n)$ worst-case update time with high probability.
This is the first dynamic MIS algorithm with very fast update time of poly-log.
\end{abstract}

\clearpage
\pagestyle{plain}
\pagenumbering{arabic}

\section{Introduction}
A maximal independent set (MIS) of a given graph $G = (V, E)$ is a subset $M$ of vertices such that $M$ does not contain two neighboring vertices and
every vertex in $V\setminus M$ has a neighbor vertex in $M$.
In this chapter, we study the maximal independent set (MIS) problem in the dynamic setting, where the graph $G$ is undergoing a sequence of edge insertions and deletions.

MIS is a fundamental problem with both theoretical and practical importance and is used as a fundamental building block in many applications.
For instance,  MIS has been used for resource scheduling for parallel threads
in a multi-core environment, for leader election \cite{Daum12},
for resource allocation \cite{Yu2014}, etc.

The MIS problem has received a lot of attention in the distributed and parallel settings since the influential works of \cite{Alon1986,Linial1987,Luby1986}.
It is considered a central problem in distributed computing and in particular in the symmetry breaking field.
Specifically, attaining a better understanding of MIS in the distributed setting is of particular interest not only because it is a fundamental problem but also because other fundamental problems reduce to it.

Censor-Hillel, Haramaty, and Karnin~\cite{censor2016optimal} in their
pioneering paper studied the problem of
maintaining an MIS in the distributed dynamic setting where the graph changes over time.
They gave a randomized algorithm for maintaining an MIS against an oblivious adversary in the distributed dynamic setting; as an open question, the authors asked whether it is possible to maintain an MIS in a dynamic graph with update time faster than recomputing everything from scratch, which triggered a recent line of research.

The first non-trivial algorithm was proposed by Assadi, Onak, Schieber and Solomon
\cite{DBLP:conf/stoc/AssadiOSS18} who presented a deterministic algorithm with $O(m^{3/4})$ amortized update time. This was the first dynamic algorithm that maintains an MIS with sublinear update time in the sequential model. This upper bound was later improved to $\tilde{O}(m^{2/3})$ independently by Du and Zhang  \cite{du2018improved} and by Gupta and Khan \cite{gupta2018simple}. In the same paper Du and Zhang \cite{du2018improved} overcame the $\tilde{O}(m^{2/3})$ barrier by assuming an oblivious adversary and a randomized algorithm with amortized update time $\tilde{O}(\sqrt{m})$ was proposed. This randomized upper bound was recently improved to $\tilde{O}(\sqrt{n})$
by Assadi \etal \cite{assadi2019fully}. For graphs with bounded arboricity $\alpha$, a deterministic algorithm with amortized update time of $O(\alpha^2\log^2 n)$ was proposed in \cite{onak2018fully}.

\subsection{Our contribution}
In this chapter we present the first dynamic MIS algorithm with very fast update time of poly-logarithmic in $n$.
We obtain a randomized dynamic MIS algorithm that works against an oblivious adversary. 
Moreover, our algorithm can actually maintain a greedy MIS with respect to a random order on the set of vertices; the concept of greedy MIS is defined as follows.
\begin{definition}
	Given any order $\pi$ on all vertices in $V$, the greedy MIS $M_\pi$ with respect to $\pi$ is uniquely defined by the following procedure that gradually builds an MIS: starting with $M_\pi =\emptyset$, for each vertex in $V$ under order $\pi$, if it is not yet dominated by $M_\pi$, add it to $M_\pi$.
\end{definition}

We say that an algorithm has worst-case expected
update time $\alpha$ if for every update $\sigma$, the expected time to process $\sigma$ is at most $\alpha$.

Our main result argues that when $\pi$ is a uniformly random permutation, the corresponding greedy MIS can be maintained under edge updates against an oblivious adversary, which is formalized in the following statement.
\begin{theorem}\label{mis-main}
	Let $\pi$ be a random permutation over $V$. The greedy MIS on $G$ according to order $\pi$ can be maintained under edge insertions and deletions in worst-case expected $O(\log^4n)$ time against an oblivious adversary, where the expectation is taken over the random choice of $\pi$.
\end{theorem}

As a corollary, we can apply a black-box reduction from worst-case time dynamic algorithms to expected worst-case time dynamic algorithms that appeared in a recent paper \cite{bernstein2019deamortization}.

\begin{theorem}[\cite{bernstein2019deamortization}]
	Let $A$ be an algorithm that maintains a dynamic data structure $D$ with worst-case expected time $\alpha$, and let $n$ be a parameter such that the maximum number of items stored in the data structure at any point in time is polynomial in $n$. Then there exists an algorithm $A^\prime$ with the following properties.
	\begin{enumerate}
		\item For any sequence of updates $\sigma_1, \sigma_2, \cdots$, $A^\prime$ processes each update $\sigma_i$ in $O(\alpha\log^2n)$ time with high probability.
		\item $A^\prime$ maintains $\Theta(\log n)$ data structures $D_1, D_2, \cdots, D_{\Theta(\log n)}$, as well as a pointer to some $D_i$ that is guaranteed to be correct at the current time. Query operations are answered with $D_i$.
	\end{enumerate}
\end{theorem}

\begin{corollary}
	There is a dynamic MIS algorithm against an oblivious adversary that handles edge updates in worst-case $O(\log^6n)$ time with high probability, and answers MIS membership queries in constant time.
\end{corollary}

\noindent\textbf{Independent work:} Independent of our work, Behnezhad \etal \cite{behnezhad2019fully} also present a fully dynamic algorithm that maintains a greedy MIS with expected poly-logarithmic running time against oblivious adversaries.

\subsection{Technical overview}
Our algorithm is a combination of techniques from \cite{censor2016optimal} and \cite{assadi2019fully}. In paper \cite{censor2016optimal}, the authors proved a lemma that the expected number of changes made to a greedy MIS by an edge update is bounded by a constant. Unfortunately, they could not achieve an efficient dynamic algorithm since a straightforward implementation of the lemma has a linear dependence on the maximum degree of the graph which could be large.

The issue with the maximum degree was overcome by the algorithm from \cite{assadi2019fully} which relies on what we informally call the \emph{degree reduction} lemma: if we pick a random subset of $k$ vertices and build a greedy MIS on this subset, then the maximum degree of the induced subgraph on all the rest un-dominated vertices is at most $O(\frac{n\log n}{k})$. Therefore we can do the following to achieve an update time with sub-linear dependence on $n$. First build an MIS on a randomly selected subset of $k$ vertices and then maintain an MIS on the induced subgraph of all the rest vertices in a brute-force manner. If an edge update lies entirely within the induced subgraph, then it takes time proportional to the maximum degree which is $\tilde{O}(n / k)$; if an edge update lies within the random subset, then we rebuild the whole data structure from scratch. The expected running time of this algorithm is a trade-off between two terms. On the one hand, when the edge update occurs within the induced subgraph, the cost would be proportional to the maximum degree which is $\tilde{O}(n / k)$; on the other hand, when the edge update connects two vertices in the random subset, the cost of rebuilding would be $O(m) = O(n^2)$, and under the assumption of obliviousness, the probability that an edge update lies within the random subset is roughly $O(\frac{k^2}{n^2})$, and so the expected time of rebuilding would be $\tilde{O}(n^2\cdot \frac{k^2}{n^2}) = \tilde{O}(k^2)$. Taking $k = \floor{n^{1/3}}$ gives a balance of $\tilde{O}(n^{2/3})$ update time. In their paper \cite{assadi2019fully}, the authors further refined the running time to $\tilde{O}(\sqrt{n})$ using a hierarchical approach.

We believe the main bottleneck of the above algorithm is that it takes no effort to utilize the lemma from \cite{censor2016optimal}. As a first attempt one could try to look for expensive parts of \cite{assadi2019fully}'s algorithm and try to plug in \cite{censor2016optimal}'s lemma. For example, instead of directly rebuilding, we could try to apply \cite{censor2016optimal}'s lemma when restoring a greedy MIS among the random subset of $k$ vertices if an edge update occurs between them. However, we would again encounter the large degree issue within the random subset.

Our new algorithm is a direct way of combining \cite{censor2016optimal}'s lemma and the degree reduction lemma. The algorithm keeps a random ordering $\pi: V\rightarrow [n]$ of all vertices and tries to maintain the random greedy MIS. In order to do so, we explicitly maintain all the induced subgraphs $G_i = (V_i, E_i)$ ($0\leq i\leq\log n$) on all vertices which are not dominated by MIS vertices from $\pi^{-1}(1), \pi^{-1}(2), \cdots, \pi^{-1}(2^i)$. For simplicity assume edge $(u, v)$ is inserted where $2^b < \pi(u) < \pi(v)\leq 2^{b+1}$ for some integer $b$. Then, on the one hand, this event happens with probability $O(2^{2b} / n^2)$ when $\pi$ is uniformly random; on the other hand, all changes to the MIS could only take place in $G_b$ whose maximum degree is bounded by $O(\frac{n\log n}{2^b})$.

Let $S\subseteq V_b$ be the set of all influenced vertices (we will formally define what $S$ is later on; basically $S$ contains all vertices that could possibly enter or leave the MIS during this update). Following similar proofs of \cite{censor2016optimal}, we could prove the conditional expectation of $S$ is at most $O(n / 2^b)$. As the maximum degree of $G_b$ is bounded by $O(\frac{n\log n}{2^b})$, we could go over all neighbors of $S$ in $G_b$ and maintain memberships of vertices from $S$ in subgraphs $G_{b+1}, G_{b+2}, \cdots$, which takes $\tilde{O}(n^2 / 2^{2b})$ time, perfectly canceling out the probability $O(2^{2b} / n^2)$ we just mentioned. However, this is not the end of the story. Not only could vertices from $S$ change their memberships in subgraphs $G_{b+1}, G_{b+2}, \cdots$, but neighbors of vertices in $S$ as well, which could be as many as $O(n^2/2^{2b})$ in the worst-case. The key to the running time analysis is that $\pi$ roughly assigns the set $S$ uniform-random positions in $[2^{b}+1, n]$ even when $S$ is given as prior knowledge. Therefore, on average, the number of neighbors in $G_b$ of a vertex in $S$ is bounded by $\tilde{O}(1)$.

\section{Preliminaries}
For any subgraph $H\subseteq G$, let $\Delta(H)$ be its maximum vertex degree. For any $U\subseteq V$, define $\Gamma(U)$ to be the set of all neighbors of $U$ in $G$, and $G[U]$ the induced subgraph of $G$ on $U$. For any permutation $\pi$ on $V$ and vertex $u\in V$, define $I_u^\pi$ to be the set of neighbor predecessors of $u$ with respect to $\pi$. For any two different vertices $u, v\in V$, we say $u$ has a \emph{higher priority} than $v$ if $\pi(u) < \pi(v)$. For any pair of indices $i, j$, define $\pi[i, j] = \{w\mid i\leq \pi(w)\leq j \}$. The following lemma states the basic characterization of a greedy MIS.
\begin{lemma}[folklore]\label{basic}
	An MIS $M$ is the greedy MIS with respect to order $\pi$ if and only if for all $z\in V$, it satisfies the constraint that either $z\in M$ or $I_z^\pi\cap M\neq \emptyset$. For the rest, we will call this constraint \emph{the greedy MIS constraint} for $z$.
\end{lemma}

The following lemma appeared in \cite{assadi2019fully}.
\begin{lemma}[\cite{assadi2019fully}]\label{deg-reduction0}
	Let $\pi$ be a uniformly random permutation on $V$ and let $k$ be an integer in $[n]$. Let $U$ be the set of all vertices not dominated by $M_\pi\cap \pi[1, k]$, then with high probability of $1 - n^{-4}$, $\Delta(G[U])\leq O(\frac{n\log n}{k})$.
\end{lemma}

The next lemma is a slight modification of the previous lemma where we show that even if we fix the position in the permutation of two vertices the lemma still holds.


\begin{lemma}\label{deg-reduction}
	Let $u_1, u_2\in V$ be two different vertices and $k_1, k_2\in [n]$ be two different indices, and let $1\leq k\leq n$ be an integer. Let $\pi$ be a uniformly random permutation on $V$ under the condition that $\pi(u_i) = k_i, i \in \{1, 2\}$. Let $U$ be the set of all vertices not dominated by $M_\pi\cap \pi[1, k]$, then with high probability $1 - n^{-2}$, $\Delta(G[U])\leq O(\frac{n\log n}{k})$.
\end{lemma}
\begin{proof}
	Call a permutation $\pi$ \emph{bad} if $\Delta(G[U])\geq \Omega(\frac{n\log n}{k})$.
	Noticing that $\Pr_\pi[\pi(u_i) = k_i, i\in \{1, 2\}] = \frac{1}{n(n-1)/2}$, by Lemma~\ref{deg-reduction0} we have:
	$$\begin{aligned}
	n^{-4}&\geq \Pr_{\pi}[\pi\text{ is bad}]\\
	&= \frac{1}{n(n-1)/2}\Pr_\pi[\pi\text{ is bad}\mid \forall i, \pi(u_i) = k_i] + (1 - \frac{1}{n(n-1)/2})\Pr_\pi[\pi\text{ is bad}\mid \exists i,\pi(u_i) \neq k_i]\\
	&\geq \frac{1}{n(n-1)/2}\Pr_\pi[\pi\text{ is bad}\mid \pi(u_i) = k_i]
	\end{aligned}$$
	Hence, $\Pr_\pi[\pi\text{ is bad}\mid\forall i, \pi(u_i) = k_i]$ is at most  $n^{-2}$ as well, which concludes the proof.
\end{proof}

For the rest of this section, we review the notion of \emph{influenced set} which was studied in \cite{censor2016optimal}. Given a total order $\pi$, an MIS $M = M_\pi$, as well as an edge update between $u, v$, we turn to define $v$'s influenced set $S_v^\pi$. If $v$ does not violate the greedy MIS constraint after the edge update, then define $S_v^\pi = \emptyset$; notice that $v$ always preserves the greedy MIS constraint if $\pi(v) < \pi(u)$. Otherwise, initially set $S_0 = \{v\}$. For any $i\geq 1$, define $S_i$ to be the set of all non-MIS vertices whose MIS predecessors are all in $S_{i-1}$, plus the set of every MIS vertex who has at least one predecessor in $S_{i-1}$, namely:
$$S_i = \{w\mid w\in M, S_{i-1}\cap I_w^\pi \neq \emptyset \}\cup \{w\mid w\notin M, I_w^\pi\cap M\subseteq \bigcup_{j=0}^{i-1}S_j \}$$

Note that the set $M$ refers to the greedy MIS in the old graph, not in the new graph. Eventually, define $v$'s influenced set to be $S_v^\pi = \bigcup_{i = 0}^\infty S_i$. When $S_v^\pi\neq\emptyset$, there is a simple characterization which will be used later.
\begin{lemma}\label{greedy-mis}
	Let $M$ be the greedy MIS in the old graph. When $S_v^\pi\neq \emptyset$, it is equal to the smallest set $S$ that contains $v$ and satisfies the following two conditions.
	\begin{enumerate}[(1)]
		\item $\forall z\in M, I_z^\pi\cap S\neq \emptyset$ iff $z\in S$.
		\item $\forall z\notin M$, $I_z^\pi\cap M\subseteq S$ iff $z\in S$.
	\end{enumerate}
\end{lemma}
\begin{proof}
	Since $S_v^\pi$ satisfies both of (1) and (2), it suffices to prove that any $S$ containing $v$ that satisfies both (1) and (2) would contain $S_v^\pi$ as an subset. This is done by an easy induction on $i\geq 0$ that $S$ contains every $S_i$.
\end{proof}

\begin{lemma}[\cite{censor2016optimal}]\label{rand-perm}
	Let $\pi, \sigma$ be two permutations, $S\subseteq V$ a nonempty set, and $v\in V$ be an arbitrary vertex. Suppose an edge update occurs between $u, v$. Assume $S_v^\pi = S$, $\pi(u) < \pi(v), \sigma(u) < \sigma(v)$, $\sigma, \pi$ have the same induced relative order on both $S$ and $V\setminus S$, namely $\pi_S = \sigma_S, \pi_{V\setminus S} = \sigma_{V\setminus S}$, then $M_\pi = M_\sigma$ in the old graph before the edge update, and $S_v^\sigma = S$.
\end{lemma}

\begin{lemma}[\cite{censor2016optimal}]\label{empty}
	Let $\pi, \sigma$ be two permutations, $S\subseteq V$ a vertex subset, and $v\in V$ be an arbitrary vertex. Suppose an edge update occurs between $u, v$. If $S_v^\pi = S\neq\emptyset$, and $\pi(u) < \pi(v), \sigma(u) < \sigma(v)$, $\sigma, \pi$ have the same induced relative order on both $S\setminus \{v\}$ and $V\setminus S$, namely $\pi_{S\setminus \{v\}} = \sigma_{S\setminus \{v\}}, \pi_{V\setminus S} = \sigma_{V\setminus S}$. If $v \neq \arg\min_{z\in S}\{\sigma(z) \}$ then $S_v^\sigma = \emptyset$.
\end{lemma}

It was also shown in \cite{censor2016optimal} that for an edge update $(u,v)$ the expected size of $S_v^\pi$ is constant. In our algorithm we need the following different variants of this claim.

\begin{lemma}\label{condition0}
	Suppose an edge update occurs between $u, v$. Let $1\leq C\leq A<B\leq n$ be three integers. Then $$\E_\pi[|S_v^\pi|\mid \pi(u)=C, \pi(v)\in [A+1, B]] < \frac{n}{B-A}$$
\end{lemma}

\begin{lemma}\label{condition}
	Suppose an edge update occurs between $u, v$. Let $1\leq A<B\leq n$ be two integers. Then $$\E_\pi[|S_v^\pi|\mid A<\pi(u) < \pi(v)\leq B] < \frac{2n}{B-A}$$
\end{lemma}
\section{The Main Algorithm}
In this section we describe our fully dynamic MIS algorithm.

\subsection{Data structure}
When $\pi$ is a fixed permutation over $V$, our algorithm is entirely deterministic. Let $M\subseteq V$ be the greedy MIS with respect to $\pi$, and for any $1\leq k\leq n$, define $M_k = M\cap \pi[1, k]$. Since $M$ is defined in a greedy manner, $M_k$ dominates the entire set $\pi[1, k]$.

The algorithm explicitly maintains the induced subgraph $G_i = (V_i, E_i), \forall 0\leq i\leq \log n-1$, where $V_i = V\setminus (M_{2^i}\cup \Gamma(M_{2^i}))$; by definition $G_0\supseteq G_1\supseteq G_2\supseteq\cdots\supseteq G_{\log n-1}$.
More precisely, given a permutation $\pi$, our algorithm maintains 
at any given point of time the graphs $G_i$ for $0\leq i\leq \log n-1$ and the greedy MIS $M_\pi$.
In the following subsection we describe our update algorithm to maintain both the graphs $G_i$ and the 
MIS $M_\pi$.

\subsection{Update algorithm}
Suppose an edge is updated, either inserted or deleted, between $u, v\in V$ with $\pi(u) < \pi(v)$. Suppose $2^a < \pi(u)\leq 2^{a+1}$ and $2^b < \pi(v) \leq 2^{b+1}$ for integers $a$ and $b$. There are several \emph{easy cases}, where $S_v^\pi = \emptyset$ and thus we do not need to make changes to $M$ as $M$ stays the greedy MIS with respect to $\pi$, and we only need to maintain the subgraphs $G_0, G_1, \cdots, G_{\log n-1}$.
\begin{enumerate}[(i)]
	\item $u\notin M$. In this case, we simply add or remove, depending whether the edge update is an insertion or deletion, the edge $(u, v)$ to/from $E_0, E_1, \cdots, E_i$, where $i$ is the largest index such that $u, v\in V_i$.
	\item $u\in M, v\notin M$, the update is a deletion and $I_v^\pi\cap M \neq \{u\}$. This case can be handled in the same way as in (i): remove the edge $(u, v)$ in $E_0, E_1, \cdots, E_i$, where $i$ is the largest index such that $u, v\in V_i$, and recompute $v$'s position in the subgraphs $G_a, G_{a+1}, \cdots, G_{\log n - 1}$.
	\item $u\in M, v\notin M$ and the update is an insertion. In this case, if $v\in V_a$, then since now $v$ is dominated by $u\in V_a$ we should remove $v$ from all subgraphs $G_k, \forall k>a$. After that, add $(u, v)$ to $E_0, E_1, \cdots, E_i$, where $i$ is the largest index such that $u, v\in V_i$.
\end{enumerate}

For the rest of this section we consider the case where an edge is inserted between $u, v\in M$, or deleted between $u\in M, v\notin M$ with $I_v^\pi\cap M = \{u\}$. In both of these cases, $S_v^\pi\neq \emptyset$ and thus we need to change $v$'s status in the MIS, and then we must try to fix the greedy MIS $M$ within $G_b$. We start by computing the nonempty influenced set $S_v^\pi$ with respect to edge update between $u, v$.
\begin{enumerate}[(1)]
	\item Initialize an output set $S = \emptyset$ that is promised to be equal to $S_v^\pi$ by the end of the algorithm, as well as a set $T = \{v\}$ that contains a set of candidate vertices that might be included in $S$ during the process.
	\item In each iteration, extract $z = \arg\min_{z\in T}\{\pi(z) \}$ from $T$. If $z\in M$, then suppose $2^k < \pi(z) \leq 2^{k+1}$; by definition it must be $z\in V_k$. First we add $z$ to $S$, and scan all neighbors $w$ of $z$ in $V_k$ such that $\pi(w) > \pi(z)$ and add $w$ to $T$.
	
	If $z\notin M$, first scan its adjacency list in $G_b$; if all its MIS neighbors with higher priority are in $S$, then add $z$ to $S$ and add all of its MIS neighbors $w\in V_b$ with $\pi(w) > \pi(z)$ to $T$.
	\item When $T$ becomes empty, output $S$ as $S_v^\pi$.
\end{enumerate}

For convenience we summarize the above procedure as pseudo-code~\ref{influence}.
\begin{algorithm}
	\caption{\textsf{FindInfluencedSet}$(u, v, b)$}
	\label{influence}
	$S\leftarrow \emptyset$, in \emph{easy cases} (\romannumeral1)(\romannumeral2)(\romannumeral3) $T\leftarrow\emptyset$, and otherwise $T\leftarrow \{v\}$\;
	\While{$T\neq \emptyset$}{
		$z\leftarrow \arg\min_{z\in T}\{\pi(z)\}$, $T\leftarrow T\setminus \{z\}$\;
		\If{$z\in M$}{
			$S\leftarrow S\cup \{z\}$\;
			suppose $2^k < \pi(z) \leq 2^{k+1}$, and assert $z\in V_k$\;
			\For{neighbor $w\in V_k$ of $z$ such that $\pi(w) > \pi(z)$}{
				$T\leftarrow T\cup \{w\}$\;
			}
		}\Else{
		$\text{flag}\leftarrow \textbf{true}$\;
		\For{neighbor $w\in V_b\cap M$ of $z$ such that $\pi(w) < \pi(z)$}{
			\If{$w\notin S$}{
				$\text{flag}\leftarrow \textbf{false}$ and \textbf{break}\;
			}
		}
		\If{flag}{
			$S\leftarrow S\cup \{z\}$\;
			\For{neighbor $w\in V_b\cap M$ of $z$ with $\pi(w) > \pi(z)$}{
				$T\leftarrow T\cup \{w\}$\;
			}
		}
	}
}
\Return $S$\;
\end{algorithm}

It will be proved that the output $S$ of Algorithm~\ref{influence} is equal to $S_v^\pi$. Once we have found $S = S_v^\pi$, we can try to fix the greedy MIS by only looking at $G[S]$; note that it might be the case that not every vertex in $S$ needs to change its status in the MIS (for example if $G[S]$ is a triangle and $v$ is removed from $M$ due to an insertion, we would not add both vertices in $S$ to $M$). If the edge update is an insertion, we first remove $v$ from all $V_k, k>a$, and then compute the greedy MIS on $G[S\setminus \{v\}]$ with respect to $\pi$; if the edge update is a deletion, we add $v$ to all $V_k, \forall a<k\leq b$, and then compute the greedy MIS on $G[S]$ with respect to $\pi$.

Last but not least, we also need to update $G_k, k\geq b+1$. This is done in the straightforward manner: go over every vertex $z$ that has changed its status in MIS in the increasing order with respect to $\pi(z)$. Assuming $2^k < \pi(z)\leq 2^{k+1}$, directly go over all of its neighbors in $G_k$ and recompute their memberships in $G_{b+1}, \cdots, G_{\log n -1}$. More specifically, consider the following two cases.
\begin{enumerate}[(1)]
	\item If $z$ has been added to $M$, then for every neighbor $w\in \Gamma(z)\cap V_k$, we remove $w$ from all $G_l, l>k$.
	\item If $z$ has been removed from $M$, then $z$ belonged to $V_k$ before the update. Instead of enumerating every neighbor from the current version of $\Gamma(z)\cap V_k$, we go over all of its old neighbors $w\in V_k$ before the update, and compute their memberships in $G_{b+1}, \cdots, G_{\log n -1}$.
\end{enumerate}

We also summarize this procedure as pseudo-code \ref{subgraph}. After that we can summarize the main update algorithm as pseudo-code \ref{update}.
\begin{algorithm}
	\caption{\textsf{FixSubgraphs}$(S, b)$}
	\label{subgraph}
	\For{$z\in S$ that has changed its status, in the increasing order in terms of $\pi$}{
		assume $2^k < \pi(z)\leq 2^{k+1}$\;
		\If{$z$ has joined $M$}{
			\For{$w\in V_k\cap \Gamma(z)$}{
				remove $w$ from all $G_l, l>k$\;
			}
		}\ElseIf{$z$ has left $M$}{
		\For{neighor $w$ of $z$ in the old version of $V_k$ before the edge update}{
			compute $w$'s memberships in $G_{k}, G_{k+1}, \cdots, G_{\log n-1}$\;
		}
	}
}
\end{algorithm}

\begin{algorithm}
	\caption{\textsf{Update}$(u, v)$}
	\label{update}
	suppose $\pi(u) < \pi(v)$, and $2^a < \pi(u)\leq 2^{a+1}$, $2^b < \pi(v)\leq 2^{b+1}$\;
	$S\leftarrow\textsf{FindInfluencedSet}(u, v, b)$\;
	\If{$S = \emptyset$}{
		recompute $v$'s memberships among $G_a, G_{a+1}, \cdots, G_{\log n -1}$\;
	}
	\Else{
		\If{the update is insertion}{
			remove $v$ from $V_k, k>a$\;
			run the greedy MIS algorithm on $G[S\setminus \{v\}]$ with respect to order $\pi$\;
		}\Else{
		add $v$ to all $V_k, a<k\leq b$\;
		run the greedy MIS algorithm on $G[S]$ with respect to order $\pi$\;
	}
	\textsf{FixSubgraphs}$(S, b)$\;
}
\end{algorithm}

\subsection{Correctness}
In this section we prove the correctness of our algorithm.
We start by proving that the algorithm correctly computes the set $S_v^\pi$.

\begin{lemma}
	Algorithm~\ref{influence} correctly outputs the influenced set with respect to $v$, namely $S = S_v^\pi$ when it terminates.
\end{lemma}
\begin{proof}
	Let $v = z_1, z_2, \cdots, z_l$ be the sequence of vertices that are added to $S$ sorted in the increasing order with respect to $\pi$. We prove inductively that for any $1\leq i\leq l$, $S_v^\pi\cap \pi[\pi(z_1), \pi(z_i)] = \{z_1, z_2, \cdots, z_i\}$. When $i = 1$, the equality holds trivially as $S_v^\pi\cap \pi[\pi(z_1), \pi(z_i)] = \{v\} = \{z_1\}$. For the inductive step, suppose we have $S_v^\pi \cap \pi[\pi(z_1), \pi(z_i)] = \{z_1, z_2, \cdots, z_i\}$ for some $i\geq 1$. Next we prove $S_v^\pi \cap \pi[\pi(z_i)+1, \pi(z_{i+1})] = \{z_{i+1}\}$ in two steps.
	\begin{itemize}
		\item $z_{i+1}\in S_v^\pi$.
		
		This can be verified according to the specification of Algorithm~\ref{influence} and definition of $S_v^\pi$ in the following way. If $z_{i+1}$ was added to $S$ on line-5, namely $z_{i+1}\in M$, then it must have been introduced to $T$ on line-17 by one of its neighbors that appeared before in $\pi$; this is because this predecessor cannot be in $M$, and so it was added to $S$ on line-15, and thus $z_{i+1}$ was added to $T$ on line-17. Then according to the definition of $S_v^\pi$, $z_{i+1}\in S_v^\pi$.
		
		If otherwise $z_{i+1}$ was added to $S$ as a non-MIS vertex, then on the one hand $z_{i+1}$ did not have MIS predecessor neighbors not in $V_b$ as $z_{i+1}\in V_b$; on the other hand, $z_{i+1}$ could be added to $S$ only when all of MIS its neighboring predecessors belong to $\{z_1, z_2, \cdots, z_i\}\subseteq S_v^\pi$. Therefore, according to the definition of $S_v^\pi$, it should be $z_{i+1}\in S_v^\pi$.
		
		\item For any $w\in \pi[\pi(z_i)+1, \pi(z_{i+1})-1]$, $w\notin S_v^\pi$.
		
		Suppose we choose $w\in S_v^\pi\cap \pi[\pi(z_i)+1, \pi(z_{i+1})-1]$ with the smallest order in $\pi$. We first rule out the case where $w\in M$. If this should be the case, the $w$ must be adjacent to a vertex $z\in \{z_1, z_2, \cdots, z_i\}$; this is not possible because $w$ would have been added to $T$, when $z$ was added to $S$ on line-15, and then later it would be added to $S$.
		
		Now we suppose $w\notin M$. By definition of $S_v^\pi$ and the inductive hypothesis, all preceding MIS neighbors of $w$ belong to $\{z_1, z_2, \cdots, z_i\}$. Let $z\in \{z_1, z_2, \cdots, z_i\}$ be the one with the smallest order among its MIS neighbors, and suppose $2^k < \pi(z)\leq 2^{k+1}$. Since $z$ is the MIS vertex that dominates $w$ with the smallest order, it must be $w\in V_k$, and therefore when $z$ was added to $S$ on line-5, $w$ would be added to $T$ on line-8, and later to $S$ on line-15, which is a contradiction.
	\end{itemize}
\end{proof}

\begin{lemma}
	Algorithm~\ref{update} correctly restores the greedy MIS with respect to $\pi$.
\end{lemma}
\begin{proof}
	We only need to consider the case when $S = S_v^\pi\neq \emptyset$ since otherwise no changes are made to the greedy MIS. We first claim that none of the vertices outside $S$ need to change their status in the greedy MIS. This is because, on the one hand, for any $z\in M\setminus S$, by Lemma~\ref{greedy-mis} we know $I_z^\pi\cap S = \emptyset$, and so any changes within $S$ cannot affect $z$; on the other hand, for any $z\in V\setminus (M\cup S)$, by Lemma~\ref{greedy-mis}, there exists a vertex from $M\setminus S$ that dominates $z$ as a predecessor, and therefore $z$ stays a non-MIS vertex, irrespective of changes in $S$.
	
	Secondly, we claim that recomputing the greedy MIS on $G[S\setminus \{v\}]$ or $G[S]$, depending on whether the update is an insertion or a deletion, has no conflicts with MIS vertices in $M\setminus S$. This is because, again by Lemma~\ref{greedy-mis}, for any $z\in S$ that was originally a non-MIS vertex, $z$ is not adjacent to any MIS vertex from $M\setminus S$, and so adding $z$ to $M$ has no conflicts with vertices in $M\setminus S$.
\end{proof}

\begin{lemma}\label{fix-subgraph}
	In each iteration of the outermost loop of Algorithm~\ref{subgraph}, by the time when line-2 is executed, $V_k$ is already equal to $V\setminus (M_{2^k}\cup \Gamma(M_{2^k}))$.
\end{lemma}
\begin{proof}
	We prove the claim by an induction on the value of $\pi(z)$. For the base case where $z = v, k = b$, note that the only possible change to $V_b$ is $v$: if the edge update is an insertion, then $v$ would leave $V_b$; if the edge update is a deletion, then $v$ might join $V_b$. In both cases, we have already fixed it right before recomputing the greedy MIS on $G[S\setminus\{v\}]$ or $G[S]$. Since we turn to fix subgraphs $G_b, G_{b+1}, \cdots, G_{\log n -1}$ after we have finished restoring the greedy MIS, it should be $V_b = V\setminus (M_{2^b}\cup\Gamma(M_{2^b}))$ at the beginning of Algorithm~\ref{subgraph}.
	
	Next we turn to look at the inductive step. We first argue that any vertex $w$ that leaves $V\setminus (M_{2^k}\cup \Gamma(M_{2^k}))$ due to changes in $S$ has already been removed from $V_k$ in previous iterations. This is because we iterate over $S$ in the increasing order with respect to $\pi$, and we must have already visited another vertex $z^\prime \in S\cap M$ with $2^l < \pi(z^\prime) \leq 2^{l+1} \leq 2^k$ who is the earliest neighbor of $w$. By the inductive hypothesis, when $z^\prime$ was enumerated in the for-loop, $V_l = V\setminus (M_{2^l}\cup \Gamma(M_{2^l}))$, and thus $w$ is removed from $V_k$ by then.
	
	Secondly we argue that any vertex $w$ that joins $V\setminus (M_{2^k}\cup \Gamma(M_{2^k}))$ due to changes in $S$ has already been added to $V_k$ in previous iterations. For $w$ to join $V\setminus (M_{2^k}\cup \Gamma(M_{2^k}))$, it must have lost all of its MIS neighbors whose order is less or equal to $2^k$. Let $z^\prime\in S\setminus M$ be the one with smallest order and assume $2^l < \pi(z^\prime) \leq 2^{l+1}\leq 2^k$, and so $z^\prime$ must have been removed from $M$ by Algorithm~\ref{update}. By the inductive hypothesis, by the time when $z^\prime$ was enumerated by Algorithm~\ref{subgraph}, we fix all old neighbors of $z^\prime$ in $V_l$, which include $w$, and hence $w$'s memberships in $G_l, G_{l+1}, \cdots, G_{\log n -1}$ were already recomputed from scratch  by then.
\end{proof}

\begin{corollary}
	The update algorithm correctly maintains subgraphs $G_0, G_1, \cdots, G_{\log n - 1}$ by the end of its execution.
\end{corollary} 
\subsection{Running time analysis}
Define $\bad$ to be the set of all permutations $\pi$ such that there exists an index $0\leq k\leq \log n - 1$ for which $\Delta(G_k)\geq \Omega(n\log n / 2^k)$ either before or after the edge update; we need to emphasize it here that the constant hidden in the $\Omega(\cdot)$ notation is larger that the constant hidden in the notation $O(\cdot)$ in the statement of Lemma~\ref{deg-reduction}.

\begin{lemma}\label{size}
	Let $a, b$ be fixed integers. Denote $\event = \{\pi(u)< \pi(v), \pi(u)\in[2^a+1, 2^{a+1}], \pi(v)\in[2^b+1, 2^{b+1}]\}$. Let $T_0$ be the set of all vertices that have once belonged to $T$, and let $T_1$ be the set of all vertices that need to change their memberships among subgraphs $G_{b+1}, \cdots, G_{\log n - 1}$ during the execution of Algorithm~\ref{subgraph}. Note that in the easy cases where $S_v^\pi = \emptyset$, we have $T = T_0 = T_1 = \emptyset$. Then we have the following bound on the conditional expectation:
	$$\E_\pi[|T_0\cup T_1|\mid \event] = O(n\log^2 n / 2^b+ n^2\cdot \Pr_\pi[\pi\in \bad\mid \event])$$
\end{lemma}

We break the proof of the above lemma into several steps.

\begin{lemma}\label{expect}
	$\E_\pi[|S_v^\pi|\mid \event] = O(n / 2^b)$.
\end{lemma}
\begin{proof}
	If $a<b$, then $\pi(u)$ belongs to $[1, 2^b]$. Directly apply Lemma~\ref{condition0} by fixing an arbitrary position $\pi(u)\in [2^a+1, 2^{a+1}]$ and setting $C = \pi(u), A = 2^b, B = 2^{b+1}$, and then we would have $\E_\pi[|S_v^\pi|\mid \event] \leq n / (2^{b+1} - A) = n / 2^b$. If $a=b$, then $\pi(u), \pi(v)\in [2^b+1, 2^{b+1}]$. Apply Lemma~\ref{condition} with $A = 2^b, B = 2^{b+1}$, and then we also have $\E_\pi[|S_v^\pi|\mid \event] \leq n / 2^{b-1}$.
\end{proof}

Fix any set $S$ such that $v\in S\subseteq V$, as well as any relative order $\sigma_+$ on $S$ and any relative order $\sigma_-$ on $V\setminus S$, such that there exists a permutation $\pi$ with $S_v^\pi = S$, $\pi_{S} = \sigma_+$, $\pi_{V\setminus S} = \sigma_-$. Therefore, we can further decompose the conditional expectations as follows:
$$\begin{aligned}
\E_\pi[|T_0\cup T_1|\mid \event]&= \sum_{S, \sigma_+, \sigma_-}\Pr_\pi[S_v^\pi = S, \pi_S = \sigma_+, \pi_{V\setminus S} = \sigma_-\mid \event]\\
&\cdot \E_\pi[|T_0\cup T_1|\mid \event, S_v^\pi = S, \pi_S = \sigma_+, \pi_{V\setminus S} = \sigma_-]
\end{aligned}$$

Therefore, it suffices to study the upper bound on $\E_\pi[|T_0\cup T_1|\mid \event, S_v^\pi = S, \pi_S = \sigma_+, \pi_{V\setminus S} = \sigma_-]$. For notational convenience, define $\Omega = \{\pi\mid \event, \pi_S = \sigma_+, \pi_{V\setminus S} = \sigma_-\}$. By Lemma~\ref{rand-perm}, if there exists one $\pi\in \Omega$ such that $S_v^\pi = S$, then all $\pi\in \Omega$ would satisfy $S_v^\pi = S$; plus $\forall \pi\in \Omega$, all $M_\pi$'s are the same in the old graph before the edge update, which we can safely denote as a common MIS $M$.

First we study the conditional expectation $\E_\pi[|T_0|\mid \pi\in \Omega]$. As can be seen from Algorithm~\ref{influence}, any vertex, which belonged to $M$ before the edge update, that has once been added to $T$ must have eventually joined $S$. So we only need to bound the total number of vertices in $T_0\setminus M$. Again by Algorithm~\ref{influence}, any $w\in T_0\setminus M$ was added to $T$ by an MIS predecessor $z\in S$ on line-8. Therefore, $|T_0\setminus M|$ is bounded by the sum of (lower priority) neighbors of all $z\in S\cap M$. So it suffices to study individual contribution of all $z\in S\cap M$ to $T_0\setminus M$. Formally, $\forall z\in S\setminus M, w\in T_0\setminus M$, we say $z$ \emph{contributes} $w$ to $T_0$ if $w$ was added to $T$ on line-8 when $z$ is being processed. First we notice a basic property of $T_0$.
\begin{lemma}\label{min-order}
	$v = \arg\min_{z\in T_0}\{\pi(z) \}$, for any $\pi\in \Omega$.
\end{lemma}
\begin{proof}
	This property is directly guaranteed by Algorithm~\ref{influence}: on line-8 or line-17, it only adds vertices $w$ to $T$ whose order is strictly larger than $z$ who has just entered $S$. Since $v$ is the first vertex that has been added to $S$, all vertices that join $T$ should have larger order than $v$.
\end{proof}

\begin{lemma}\label{relative}
	For any $k>b$, $\E_{\pi\in\Omega}[|(S\setminus\{v\})\cap \pi[2^b+1, 2^{k}]|] < \frac{2^k|S|}{n}$.
\end{lemma}
\begin{proof}
	Decompose the expectation as following:
	$$\begin{aligned}
	&\E_{\pi\in\Omega}[|(S\setminus\{v\})\cap \pi[2^b+1, 2^{k}]|]\\
	&= \sum_{j=2^b+1}^{2^{b+1}}\Pr_{\pi\in\Omega}[\pi(v) = j]\cdot \E_{\pi\in\Omega}[|(S\setminus\{v\})\cap \pi[2^b+1, 2^{k}]|\mid \pi(v) = j]
	\end{aligned}$$
	
	When $\pi(v) = j$, the rest of $S\setminus \{v\}$ are free to choose positions on $[j+1, n]$, as $v$ always takes the smallest order among $S$, which is guaranteed by Lemma~\ref{min-order} as $S\subseteq T_0$. Hence, for any $l\in[1, \min\{2^k-j, |S|-1\}]$, conditioned on $\pi(v) = j$, the probability that $|(S\setminus\{v\})\cap\pi[2^b+1, 2^k]|=l$ is equal to $\frac{\binom{2^k-j}{l}\cdot\binom{n-2^k}{|S|-1-l}}{\binom{n-j}{|S|-1}}$. Therefore, the expectation is computed as follows:
	$$\begin{aligned}
	&\E_{\pi\in\Omega}[|(S\setminus\{v\})\cap\pi[2^b+1, 2^k]|\mid \pi(v) = j]\\
	&= \sum_{l=1}^{\min\{2^k-j, |S|-1\}}l\cdot \Pr_{\pi\in\Omega}[|(S\setminus\{v\})\cap\pi[2^b+1, 2^k]| = l\mid \pi(v) = j]\\
	&= \sum_{l=1}^{\min\{2^k-j, |S|-1\}}l\cdot\frac{\binom{2^k-j}{l}\cdot\binom{n-2^k}{|S|-1-l}}{\binom{n-j}{|S|-1}}\\
	&= \sum_{l=1}^{\min\{2^k-j, |S|-1\}}(2^k-j)\cdot\binom{2^k-j-1}{l-1}\cdot\binom{n-2^k}{|S|-1-l} / \binom{n-j}{|S|-1}\\
	&= (2^k-j)\cdot\binom{n-j-1}{|S|-2} / \binom{n-j}{|S|-1}=\frac{(2^k-j)(|S|-1)}{n-j} < \frac{2^k|S|}{n}
	\end{aligned}$$
	
	Finally, we have:
	$$\begin{aligned}
	\E_{\pi\in\Omega}[|(S\setminus\{v\})\cap \pi[2^b+1, 2^{k}]|] &= \sum_{j=2^b+1}^{2^{b+1}}\Pr_{\pi\in\Omega}[\pi(v) = j]\\
	&\cdot \E_{\pi\in\Omega}[|(S\setminus\{v\})\cap \pi[2^b+1, 2^{k}]|\mid \pi(v) = j]\\
	&< \sum_{j=2^b+1}^{2^{b+1}}\Pr_{\pi\in\Omega}[\pi(v) = j]\cdot \frac{2^k|S|}{n} = \frac{2^k|S|}{n}
	\end{aligned}$$
\end{proof}

\begin{lemma}\label{T0}
	The expected contribution of all $z\in S\cap M\setminus \{v\}$ to $T_0$ is at most $O(|S|\log^2n + |S|n\cdot\Pr_{\pi\in\Omega}[\pi\in\bad])$.
\end{lemma}
\begin{proof}
	Consider any index $b\leq k\leq \log n - 1$. When $2^k < \pi(z)\leq 2^{k+1}$, the total number of neighbors of $z$ in $V_k$ is at most $O(n\log n / 2^k + n\cdot \ind[\pi\in\bad])$, by definition of $\bad$. Therefore, by Lemma~\ref{relative} the expected total contribution of $z\in S\cap M\setminus \{v\}$ to $T_0$ that lies in $[2^k+1, 2^{k+1}]$ is bounded by $O(|S|\log n + 2^k|S|\cdot \ind[\pi\in\bad])$. Taking a summation over all $k$ we can finalize the proof.
\end{proof}

By Lemma~\ref{T0}, we have an upper bound on conditional expectation:
$$\E_{\pi}[|T_0|\mid \pi\in\Omega]\leq O(|S|\log^2 n + n\log n / 2^b + |S|n\cdot \Pr_{\pi\in\Omega}[\pi\in\bad])$$
Here we have an extra additive term as an upper bound on the contribution of $v$ to $T_0$.

Next we try to study $\E_{\pi}[|T_1|\mid \pi\in\Omega]$. By Algorithm~\ref{subgraph}, for any $z\in S$ that has changed its status in $M$, we go over some of the neighbors of $z$ and update their memberships in $G_{k+1}, \cdots, G_{\log n - 1}$ using brute force, and by definition these neighbors would belong to $T_1$. Similar to what we did with $T_0$, we say $z$ \emph{contributes} these neighbors to $T_1$. Next we need to carefully analyze the total number of these neighbors.

\begin{lemma}\label{T1}
	The expected contribution of all $z\in S\setminus \{v\}$ to $T_1$ is at most $O(|S|\log^2n + |S|n\cdot\Pr_{\pi\in\Omega}[\pi\in\bad])$.
\end{lemma}
\begin{proof}
	Let $k\in [b, \log n-1]$ be any index. Assume $2^k<\pi(z)\leq 2^{k+1}$. Consider the following two possibilities.
	\begin{itemize}
		\item $z$ has joined $M$ during the update algorithm. In this case, $z$ must belong to $V_k$ and thus the total number of its neighbors in $V\setminus (M_{2^k}\cup \Gamma(M_{2^k}))$ is at most $O(n\log n / 2^k + n\cdot\ind[\pi\in\bad])$, and by Lemma~\ref{fix-subgraph} we know $V_k = V\setminus (M_{2^k}\cup \Gamma(M_{2^k}))$ by the time $z$ is processed by Algorithm~\ref{subgraph}, and thus the total number of its neighbors in $V_k$ is at most $O(n\log n / 2^k + n\cdot\ind[\pi\in\bad])$.
		\item $z$ has just left $M$ during the update algorithm. In this case, $z$ was selected by $M$ and thus belonged to $V_k$ before the update. As Algorithm~\ref{subgraph} only iterates over $z$'s old neighbors in $V_k$, the total number of these neighbors is also bounded by $O(n\log n / 2^k + n\cdot\ind[\pi\in\bad])$.
	\end{itemize}
	In any case, the contribution of $z$ to $T_1$ is at most $O(n\log n / 2^k + n\cdot\ind[\pi\in\bad])$. Therefore, by Lemma~\ref{relative} the expected total contribution of $z\in S\cap M\setminus \{v\}$ to $T_1$ that lies in $[2^k+1, 2^{k+1}]$ is bounded by $O(|S|\log n + 2^k|S|\cdot \ind[\pi\in\bad])$. Taking a summation over all $k$ we can finalize the proof.
\end{proof}

Taking a summation over all $z\in S\setminus \{v\}$ that has changed its status in the MIS we have:
$$\E_{\pi}[|T_1|\mid \pi\in\Omega]\leq O(|S|\log^2 n + n\log n / 2^b + |S|n\cdot \Pr_{\pi\in\Omega}[\pi\in\bad])$$
Here the extra additive term also stands for $v$'s contribution to $T_1$.

\begin{proof}[Proof of Lemma~\ref{size}]
	To summarize, by Lemma~\ref{T0} and Lemma~\ref{T1}, we have proved:
	$$\E_\pi[|T_0\cup T_1|\mid \pi\in \Omega ]\leq O(|S|\log^2n + n\log n / 2^b + |S|n\cdot \Pr_{\pi\in\Omega}[\pi\in\bad])$$
	
	Recall a previous decomposition we would then have:
	$$\begin{aligned}
	\E_\pi[|T_0\cup T_1|\mid \event]&= \sum_{S, \sigma_+, \sigma_-}\Pr_\pi[S_v^\pi = S, \pi_S = \sigma_+, \pi_{V\setminus S} = \sigma_-\mid \event]\\
	&\cdot \E_\pi[|T_0\cup T_1|\mid \event, S_v^\pi = S, \pi_S = \sigma_+, \pi_{V\setminus S} = \sigma_-]\\
	&\leq \sum_{S, \sigma_+, \sigma_-}O(|S|\log^2n + n\log n / 2^b\\
	& + |S|n\cdot \Pr_{\pi\in\Omega}[\pi\in\bad])\cdot\Pr_\pi[S_v^\pi = S, \pi_S = \sigma_+, \pi_{V\setminus S} = \sigma_-\mid \event]\\
	&= \sum_{S}O(|S|\log^2n\cdot\Pr_\pi[S_v^\pi = S\mid \event]) + O(n\log n / 2^b )\\
	& +\sum_{S, \sigma_+, \sigma_-} |S|n\cdot \Pr_{\pi}[\pi\in\bad\mid\event,S_v^\pi = S, \pi_S = \sigma_+, \pi_{V\setminus S} = \sigma_-]\\
	&\cdot\Pr_\pi[S_v^\pi = S, \pi_S = \sigma_+, \pi_{V\setminus S} = \sigma_-\mid \event]\\
	&\leq O(\E_\pi[|S_v^\pi|\log^2n\mid \event] + n\log n/2^b + n^2\Pr_\pi[\pi\in\bad\mid \event])\\
	& \leq O(n\log^2n / 2^b + n^2\Pr_\pi[\pi\in\bad\mid \event])
	\end{aligned}$$
	The last inequality holds by Lemma~\ref{expect}.
\end{proof}

To remove the extra term $\Pr_\pi[\pi\in \bad\mid \event]$, apply Lemma~\ref{deg-reduction} by fixing values of $\pi(u), \pi(v)$ and taking union bound over all $k$ equal to powers of $2$, we would know that $\pi\notin\bad$ with high probability, namely $\Pr_\pi[\pi\in\bad\mid \event] \leq n^{-2}\log n $, and thus $\E_\pi[|T_0\cup T_1|\mid \event]\leq O(n\log^2n / 2^b +\log n) = O(n\log^2n / 2^b)$.

By definition of $T_0$ and $T_1$, the total update time is proportional to $\Delta(G_b)\cdot (|T_0| + |T_1|)$ whose expectation is then bounded by $O(n^2\log^3n / 2^{2b})$. Since fixing the memberships of $v$ takes time at most $O(n\log^2n / 2^a)$, it immediately says that the expected update time is $O(n^2\log^3n / 2^{2b} + n\log^2n / 2^a)$. Since the adversary is oblivious to the randomness used in the algorithm, the probability of inserting an edge between $\pi[2^a+1, 2^{a+1}]$ and $\pi[2^b+1, 2^{b+1}]$ is $O(2^{a+b} / n^2)$. Hence, the expected running time would be $O(2^{a+b} / n^2\cdot (n^2\log^3n / 2^{2b} + n\log^2n / 2^a)) = O(2^{a-b}\log^3 n + \log^2n)$. Summing over all different indices $a, b$, the total time would be $O(\log^4 n)$. 

\section*{Acknowledgement}
The second author would like to thank Hengjie Zhang for helpful discussions.

\vspace{5mm}
\bibliographystyle{plain}
\bibliography{ref}

\appendix
\section{Missing Proofs}

\subsection{Proof of Lemma \ref{rand-perm}}
\begin{proof}
	Here we present an conceptually simpler proof than the one presented in \cite{censor2016optimal}. Notice that it suffices to consider the case where $\sigma(z) = \pi(z), \forall z\notin \{x, y\}$ and $\sigma(x) = \pi(y), \sigma(y) = \pi(x)$, where $x, y$ is an arbitrary pair of consecutive vertices in $\pi$ such that $x\in S$ and $y\notin S$. As $\sigma(u) < \sigma(v), \pi(u) < \pi(v)$, it can never be $x = v$ and $y = u$. Let $M = M_\pi$ be the greedy MIS in the old graph. The proof follows from the two statements below.
	\begin{claim}
		$M$ was also the greedy MIS on the old version of $G$ with respect to $\sigma$.
	\end{claim}
	\begin{proof}[Proof of claim]
		Recall from Lemma~\ref{basic}, $M$ is the greedy MIS with respect to $\sigma$ in the old graph if $M$ is an MIS and for all $z\in V\setminus M$, $I_z^\sigma\cap M\neq \emptyset$. The first half is easy: $M$ was the greedy MIS in the old version of $G$ with respect to $\pi$, so $M$ is certainly an MIS in the old graph. Now we turn to verify the greedy MIS constraints.
		
		Since $\sigma$ agrees with $\pi$ on every vertex except for $\{x, y\}$, we only need to verify $\forall z\in \{x, y\}\setminus M$, $I_z^\sigma\cap M\neq \emptyset$. We can assume $x, y$ are neighbors in the updated graph $G$; otherwise switching the orders between $x, y$ in $\pi$ does not affect the greedy MIS constraint. Consider several cases.
		\begin{itemize}
			\item $x\in M, y\notin M$. In this case, if $\pi(y) > \pi(x)$, then $\sigma(x) > \sigma(y)$, and thus $x\in I_y^\sigma\cap M\neq\emptyset$. If $\pi(y) < \pi(x)$, then since $y\notin S$, $I_y^\pi\cap M\setminus S$ must be nonempty, and so there exists $z\in I_y^\pi\cap M\setminus S$ that dominates $y$. As $\sigma(z) = \pi(z) < \pi(x) = \sigma(y)$, $I_y^\sigma \cap M$ is also nonempty.
			
			\item $x, y\notin M$. Since $x, y$ are consecutive in $\pi$, switching their positions in $\sigma$ does not affect the invariant that $I_z^\sigma\cap M\neq \emptyset, \forall z\in \{x, y\}$.
			
			\item $x\notin M, y\in M$. By definition of $S$, $\pi(y) < \pi(x)$ as otherwise $y$ would belong to $S$, and so $\sigma(y) > \sigma(x)$. If $x\neq v$, then $x$ cannot belong to $S$ by definition since $x$ is dominated by some MIS vertices outside of $S$. If $x = v$, then $y\neq u$ as $\sigma(v) > \sigma(u)$. Right after the edge update $x$ is still dominated by a vertex in $M$, namely $y$, which is also a predecessor in $\pi$, so $S = \emptyset$ which is a contradiction.
		\end{itemize}
	\end{proof}
	\begin{claim}
		$S_v^\sigma = S$.
	\end{claim}
	\begin{proof}[Proof of claim]
		By the previous claim, $M$ was also the greedy MIS on $G$ with respect to order $\sigma$. We first argue that $S_v^\sigma\supseteq S$. To do this, we prove by an induction that for every $i\geq 0$, $S_i\subseteq S_v^\sigma$; we refer readers to the definition of influenced sets for the meaning of $S_i$, where $S_i$'s are defined with respect to permutation $\pi$, not $\sigma$.
		\begin{itemize}
			\item Basis. For $i = 0$, to argue $v\in S_v^\sigma$ we only need to prove $S_v^\sigma\neq \emptyset$. As $S_v^\pi\neq\emptyset$, the edge update can only be an insertion $(u, v)$ and $u, v\in M$, or an edge deletion $(u, v)$ and $u\in M, v\notin M$ plus that $u$ is the only MIS predecessor that dominates $v$. Since $\sigma$ and $\pi$ agree on all vertices whose orders are $\leq \pi(v)$, $v$ would also violate its greedy MIS constraint with respect to $\sigma$, and so $S_v^\sigma\neq\emptyset$.
			\item Induction. Suppose we already have $S_{i-1}\subseteq S_v^\sigma$. Then, by Lemma~\ref{greedy-mis}, any $z\in M$ such that $S_{i-1}\cap I_z^\sigma\neq\emptyset$ should belong to $S_v^\sigma$. Since $\pi$ and $\sigma$ have the same relative order on $S$, $S_{i-1}\cap I_z^\sigma$ would be the same as $S_{i-1}\cap I_z^\pi$ for any $z\in S_i\cap M$. On the other hand, for any $z\in S_i\setminus (M\cup \{v\})$, we claim $I_z^\sigma\cap M$ is also equal to $I_z^\pi\cap M$. The only possible violation comes from the case that $z=x$ and $y\in M$. However this is also not possible: if $\pi(y) > \pi(x)$, then as $y\notin S$, by definition when $x\neq u$, it would have been excluded from $S$, and otherwise if $x=v$ we would have $S_v^\pi=\emptyset$; if $\pi(y) < \pi(x)$, then $y$ would have been added to $S$; both lead to contradictions.
			
			Therefore, by definition of $S_i$, we also have $S_i\subseteq S_v^\sigma$.
		\end{itemize}
		
		To prove $S_v^\sigma\subseteq S$, by Lemma~\ref{greedy-mis} it suffices to verify that (1) $\forall z\in M$, $I_z^\sigma\cap S \neq \emptyset$ iff $z\in S$; (2) $\forall z\notin M$, $I_z^\sigma\cap M\subseteq S$ iff $z\in S$. As $\sigma$ is equal to $\pi$ except for $x, y$, we only need to consider $z\in \{x, y\}$ in (1)(2). We can assume $x, y$ are adjacent; otherwise switching the orders between $x, y$ in $\pi$ does not affect the invariant. Then it can never be the case where $x\notin M, y\in M$ as it would contradict the definition of $S$. So it is either $x\in M, y\notin M$ or $x, y\notin M$. Consider two cases.
		\begin{itemize}
			\item $x\in M, y\notin M$. In this case, $I_x^\sigma\cap S=\emptyset$ always holds as switching the positions between $x, y$ does not affect the equality $I_x^\sigma\cap S = I_x^\pi\cap S \neq\emptyset$.
			
			If $\pi(y) < \pi(x)$, then since $y\notin M$, it must be $I_z^\pi\cap M\neq \emptyset$, and because $y\notin S$, there exists $z\in I_z^\pi\cap M\setminus S$. So $\sigma(z) = \pi(z) < \pi(y) = \sigma(x)$. By the previous claim we already know $M_\sigma = M$, and so $I_z^\sigma\cap M\nsubseteq S$. If $\pi(y) > \pi(x)$, then $I_z^\sigma\cap S\subseteq I_z^\pi\cap S = \emptyset$.
			
			\item $x, y\notin M$. Since $x, y$ are consecutive in $\pi$, switching their positions in $\sigma$ does not change $I_z^\sigma\cap M, \forall z\in\{x, y\}$.
			
			\item $x\notin M, y\in M$. By definition of $S$, $\pi(y) < \pi(x)$ as otherwise $y$ would belong to $S$, and so $\sigma(y) > \sigma(x)$. If $x\neq u$, then $x$ cannot belong $S$ by definition since $x$ is dominated by some MIS vertices outside of $S$. If $x = v$, then $y\neq u$ as $\sigma(v) > \sigma(u)$. Right after the edge update $x$ is still dominated by a vertex in $M$, namely $y$, which is also a predecessor in $\pi$, so $S = \emptyset$ which is a contradiction.
		\end{itemize}
	\end{proof}
\end{proof}

\subsection{Proof of Lemma \ref{empty}}
\begin{proof}
	As the lemma is stated in a slightly different way from \cite{censor2016optimal}, for completeness we also present a proof here. Define an intermediate permutation $\tau$ by this operation: remove $v$ from order $\sigma$ and reinsert it back right after $u$. Then $\tau(u) < \tau(v), \tau_S = \pi_S, \tau_{V\setminus S} = \pi_{V\setminus S}$, and thus by Lemma~\ref{rand-perm} we have $S_v^{\tau} = S$. Namely, $\tau$ and $\pi$ satisfy the same condition in the statement of the lemma.
	
	Let $w = \arg\min_{x\in S\setminus \{v\}}\{\tau(x)\}$. First we argue that $w$ and $v$ are neighbors. If $w$ was in $M_\tau$, then by the inductive definition of $S_v^\tau$, there exists $z\in S\setminus M_\tau$ such that $z$ is a predecessor neighbor of $w$. By minimality of $w$, $z$ must be equal to $v$, and hence $w$ and $v$ are adjacent. If $w$ was not in $M_\tau$, then it has an MIS predecessor $z\in S\cap M_\tau$, similarly by minimality of $w$, $z$ must be equal to $v$, and hence $w$ and $v$ are adjacent.
	
	Recalling the relation between $\tau$ and $\sigma$, we can view $\sigma$ as a permutation derived from $\tau$ by first removing $v$ from $\tau$ and then reinsert $v$ back to $\tau$ at a certain position somewhere behind $w$. We claim that right after we remove $v$ from $\tau$ before reinsertion, $w$ belongs to the greedy MIS $M_\tau$ with respect to the current $\tau$ (which is without $v$). Consider the only two cases where $S_v^\tau$ could be nonempty.
	\begin{itemize}
		\item The edge update is an insertion and both of $u, v$ were in $M_\tau$. After the removal of $v$, $w$ is no longer dominated by any MIS predecessor in $M_\tau$, hence $w$ must join $M_\tau$.
		\item The edge update is a deletion, and $u$ was in $M_\tau$ while $v$ was not in $M_\tau$, plus that $u$ is the only MIS predecessor that dominates $v$. Since $v$ was not in $M_\tau$, then by minimality of $\tau(w)$ among $S\setminus \{v\}$, the only predecessor of $w$ in $S$ was $v$, and thus $w\in M_\tau$ before and after $v$'s removal.
	\end{itemize}
	
	When we insert $v$ back to $\tau$ at some position after $w$, which produces permutation $\sigma$, since $w$ is now an MIS predecessor of $v$, $v$ does not belong to $M_\sigma$. If the edge update is insertion then no changes would be made to $M_\sigma$ and thus $S_v^\sigma = \emptyset$; if the edge update is deletion, then since $v$ has a neighboring MIS predecessor other than $u$, which is $w$, $M_\sigma$ would also stay unchanged, and thus $S_v^\sigma = \emptyset$.
\end{proof}

\subsection{Proof of Lemma \ref{condition0}}
%
%
\begin{proof}
	For notational convenience, define $\event = \{\pi(u)=C, \pi(v)\in [A+1, B]\}$. For any vertex set $S$ containing $v$, and partial orders $\sigma_+, \sigma_-$ on $S\setminus \{v\}$ and $V\setminus S$, with the property that there exists at least one permutation $\pi$ that satisfies event $\event$, as well as $S_v^\pi = S, \pi_{S\setminus \{v\}} = \sigma_+, \pi_{V\setminus S} = \sigma_-$, define a set of permutations
	$$\Omega_{S, \sigma_+, \sigma_-}= \{\pi\mid \event, \pi_{S\setminus \{v\}} = \sigma_+, \pi_{V\setminus S} = \sigma_-, \min_{z\in S}\{\pi(z)\}>A  \}$$
	
	By Lemma~\ref{rand-perm} and Lemma~\ref{empty}, for any $\pi\in\Omega_{S, \sigma_+, \sigma_-}$, we have $S_v^\pi = S$ when $\pi(v) = \min_{z\in S}\{\pi(z)\}$, and $S_v^\pi = \emptyset$ otherwise. Here is a basic property of $\Omega_{S, \sigma_+, \sigma_-}$.
	
	\begin{claim}
		For any two different sets 
		$$\Omega_{S, \sigma_+, \sigma_-} = \{\pi\mid \event, \pi_{S\setminus \{v\}} = \sigma_+, \pi_{V\setminus S} = \sigma_-, \min_{z\in S}\{\pi(z)\}>A  \}$$ and $$\Omega_{S^\prime, \sigma_+^\prime, \sigma_-^\prime} = \{\pi\mid \event, \pi_{S^\prime\setminus \{v\}} = \sigma_+^\prime, \pi_{V\setminus S^\prime} = \sigma_-^\prime, \min_{z\in S}\{\pi(z)\}>A  \}$$
		we claim $\Omega_{S, \sigma_+, \sigma_-}$ and $\Omega_{S^\prime, \sigma_+^\prime, \sigma_-^\prime}$ are disjoint.
	\end{claim}
	\begin{proof}[Proof of claim]
		Suppose otherwise there exists $\tau\in \Omega_{S, \sigma_+, \sigma_-}\cap\Omega_{S^\prime, \sigma_+^\prime, \sigma_-^\prime}$. By definition, there exists $\pi\in\Omega_{S, \sigma_+, \sigma_-}$ that satisfies event $\event$, as well as $S_v^\pi = S, \pi_{S\setminus \{v\}} = \sigma_+, \pi_{V\setminus S} = \sigma_-$. By Lemma~\ref{empty}, $v$ takes the minimum in $\pi$ among $S$.
		
		Remove $v$ from $\tau$ and reinsert $v$ back to $\tau$ right at position $A+1$. We claim $\tau$ stays in $\Omega_{S, \sigma_+, \sigma_-}\cap\Omega_{S^\prime, \sigma_+^\prime, \sigma_-^\prime}$; this is because removal and reinsertion of $v$ preserves $\tau$'s induced order on $S\setminus \{v\}, V\setminus S$ and $S^\prime\setminus \{v\}, V\setminus S^\prime$. Now, since $v$ takes the minimum among $S$ in $\tau$, we have $\tau_S = \pi_S, \tau_{V\setminus S} = \pi_{V\setminus S}$. By Lemma~\ref{rand-perm}, $S_v^\tau = S_v^\pi = S$. Similarly we can also have $S_v^\tau = S^\prime$. Therefore, $S = S^\prime$. As $\tau\in \Omega_{S^\prime, \sigma_+^\prime, \sigma_-^\prime}$, we know immediately $\sigma_+ = \tau_S = \tau_{S^\prime} = \sigma_+^\prime, \sigma_- = \tau_{V\setminus S} = \tau_{V\setminus S^\prime} = \sigma_-^\prime$, which is a contradiction that $\Omega$ and $\Omega^\prime$ are different.
	\end{proof}
	
	By this claim, we can decompose the expectation as a sum of conditional ones:
	$$\E_\pi[|S_v^\pi|\mid \event] = \sum_{S, \sigma_+, \sigma_-}\Pr_{\pi}[\pi\in\Omega_{S, \sigma_+, \sigma_-}\mid\event]\cdot\E_\pi[|S_v^\pi|\mid \event, \pi\in \Omega_{S, \sigma_+, \sigma_-}]$$
	
	So it suffices to compute each term in the summation. Fix any $S, \sigma_+, \sigma_-$ and $\Omega = \Omega_{S, \sigma_+, \sigma_-}$
	Notice that by Lemma~\ref{rand-perm} and Lemma~\ref{empty} we have:
	$$\begin{aligned}
	\E_\pi[|S_v^\pi|\mid \event, \pi\in \Omega] &= |S|\cdot \Pr_\pi[\pi(v) = \min_{z\in S}\{\pi(z)\} \mid \event, \pi\in \Omega]
	\end{aligned}$$
	
	To bound the probability $\Pr_\pi[\pi(v) = \min_{z\in S}\{\pi(z)\} \mid \event, \pi\in \Omega]$, on the one hand, any permutation $\pi\in \Omega$ can be constructed by picking an arbitrary position for $v$ among $[A+1, B]$, and then assign arbitrary positions for $S\setminus \{v\}$, so $|\Omega| = (B-A)\cdot\binom{n-A-1}{|S|-1}$. On the other hand, the total number of permutations such that $v$ takes the minimum among $S$ is $\binom{n-A}{|S|} - \binom{n-B}{|S|}$. Therefore, as $\pi$ is uniformly drawn from $\Omega$, we have:
	$$\begin{aligned}
	\Pr_\pi[\pi(v) = \min_{z\in S}\{\pi(z)\} \mid \event, \pi\in \Omega] &= \frac{\binom{n-A}{|S|} - \binom{n-B}{|S|}}{(B-A)\cdot\binom{n-A-1}{|S|-1}}\\
	&= \frac{\binom{n-A}{|S|} - \binom{n-B}{|S|}}{(B-A)\cdot\binom{n-A}{|S|}\cdot\frac{|S|}{n-A}} < \frac{n-A}{(B-A)|S|}
	\end{aligned}$$
	
	Hence, $\E_\pi[|S_v^\pi|\mid \event, \pi\in \Omega] = |S|\cdot \Pr_\pi[\pi(v) = \min_{z\in S}\{\pi(z)\} \mid \event, \pi\in \Omega]< \frac{n-A}{B-A}$. Since all $\Omega$ are disjoint, ranging over all different choices for $S, \sigma_+, \sigma_-$, we have $$\E_\pi[|S_v^\pi|\mid \pi(u)=C, \pi(v)\in [A+1, B]] < \frac{n-A}{B-A} < \frac{n}{B-A}$$
\end{proof}

\subsection{Proof of Lemma \ref{condition}}
%
\begin{proof}
	For $u, v$ to both lie in $[A+1, B]$, $B$ must be larger than $A+1$. For notational convenience, define $\event = \{A<\pi(u) < \pi(v)\leq B\}$. We decompose the expectation as:
	$$\begin{aligned}
	\E_\pi[|S_v^\pi|\mid \event] &= \sum_{k=A+1}^{B-1} \Pr_\pi[\pi(u) = k\mid \event]\cdot \E_\pi[|S_v^\pi|\mid \event, \pi(u) = k]\\
	&= \sum_{k=A+1}^{B-1} \frac{B-k}{\binom{B-A}{2}}\cdot \E_\pi[|S_v^\pi|\mid \event, \pi(u) = k]
	\end{aligned}$$
	The second equality holds as $\Pr_\pi[\pi(u) = k\mid \event] = \frac{B-k}{\binom{B-A}{2}}$; this is because, conditioned on $\pi(u) = k$ as well as event $\event$, there are $(B-k)\cdot (n-A-2)!$ permutations $\pi$, while there are $\binom{B-A}{2}\cdot (n-A-2)!$ permutations $\pi$ that satisfy event $\event$. Since $\pi$ is drawn uniformly at random from the set of all permutations that satisfy event $\event$, we have $\Pr_\pi[\pi(u) = k\mid \event] = \frac{B-k}{\binom{B-A}{2}}$.
	
	Using Lemma~\ref{condition0}, we have:
	$$\begin{aligned}
	\E_\pi[|S_v^\pi|\mid \event, \pi(u) = k]\leq\frac{n}{B-k}
	\end{aligned}$$
	
	Therefore,
	$$\begin{aligned}
	\E_\pi[|S_v^\pi|\mid \event]
	&= \sum_{k=A+1}^{B-1} \frac{B-k}{\binom{B-A}{2}}\cdot \E_\pi[|S_v^\pi|\mid \event, \pi(u) = k]\\
	&< \sum_{k=A+1}^{B-1}\frac{B-k}{(B-A-1)(B-A)/2}\cdot\frac{n}{B-k} < \frac{2n}{B-A}
	\end{aligned}$$
\end{proof}

\end{document}